\documentclass[11pt,letter]{article}

\usepackage{cite}
\usepackage{graphicx}
\usepackage{amsmath}
\usepackage{amsthm}
\usepackage{amsfonts}
\usepackage{amssymb}
\usepackage{fullpage}
\usepackage[margin=1in]{geometry}
\usepackage{color}
\usepackage{latexsym}
\usepackage{algorithm}
\usepackage[noend]{algorithmic}
\usepackage{subfigure}
\usepackage{soul}

\newtheorem{theorem}{Theorem}[section]
\newtheorem{corollary}[theorem]{Corollary}
\newtheorem{lemma}[theorem]{Lemma}

\newtheorem{proposition}[theorem]{Proposition}
\newtheorem{problem}[theorem]{Problem}

\newcommand{\old}[1]{{{}}}

\def\segment#1#2{{\overline{#1#2}}}
\def\wedge#1{{W_{{#1}}}}
\def\halfplane#1{{h_{{#1}}}}
\def\leftray#1{{\rho{\stackrel{_\nwarrow}{_{#1}}}}}
\def\rightray#1{{\rho{\stackrel{_\nearrow}{_{#1}}}}}
\def\leftline#1{{l(\leftray{#1})}}
\def\rightline#1{{l(\rightray{#1})}}
\def\topregion#1{{{R}^{top}_{{#1}}}}
\def\bottomregion#1{{{R}^{bot}_{{#1}}}}
\def\clst#1{{\C_{{{F}}({{#1}})}}}

\def\C{{\cal{C}}}

\def\grid{{\cal{G}}}
\def\graph{{{G}}}
\def\UDG{{U\!DG}}
\def\nf{{n\!f}}

\begin{document}

\title{Symmetric Connectivity with Directional Antennas\thanks{Work by R. Aschner and G. Morgenstern was partially supported by the Lynn and William Frankel Center for Computer Sciences. Work by R. Aschner and M. Katz was partially supported by the Israel Ministry of Industry, Trade and Labor (consortium CORNET). Work by M. Katz was partially supported by grant 1045/10 from the Israel Science Foundation.
}}

\author{Rom Aschner \ \ \ Matthew J. Katz \ \ \ Gila Morgenstern
\\
\\
{\small Department of Computer Science, Ben-Gurion University, Israel} \\
{\small {\tt $\{$romas,matya,gilamor$\}$@cs.bgu.ac.il}}}

\old{
\author{Rom Aschner\thanks{Partially supported by the Lynn and William Frankel Center for Computer Sciences, and by the Israel Ministry of Industry, Trade and Labor (consortium CORNET).} \ \ \ Matthew J. Katz\thanks{Partially supported by grant 1045/10 from the Israel Science Foundation, and by the Israel Ministry of Industry, Trade and Labor (consortium CORNET).} \ \ \ Gila Morgenstern\thanks{Partially supported by the Lynn and William Frankel Center for Computer Sciences.}
\\
{\small Department of Computer Science, Ben-Gurion University, Israel} \\
{\small {\tt $\{$romas,matya,gilamor$\}$@cs.bgu.ac.il}}}
}
\old{
\author{Rom Aschner\thanks{Dept. of Computer Science, Ben-Gurion University, Beer-Sheva 84105, Israel, {\tt romas@cs.bgu.ac.il}. Partially supported by the Lynn and William Frankel Center for Computer Sciences, and by the Israel Ministry of Industry, Trade and Labor (consortium CORNET).} \and Matthew J. Katz\thanks{Dept. of Computer Science, Ben-Gurion University, Beer-Sheva 84105, Israel, {\tt matya@cs.bgu.ac.il}. Partially supported by grant 1045/10 from the Israel Science Foundation, and by the Israel Ministry of Industry, Trade and Labor (consortium CORNET).} \and Gila Morgenstern\thanks{Dept. of Computer Science, Ben-Gurion University, Beer-Sheva 84105, Israel, {\tt gilamor@cs.bgu.ac.il}. Partially supported by the Lynn and William Frankel Center for Computer Sciences.}}
}

\maketitle


\begin{abstract}
Let $P$ be a set of points in the plane, representing transceivers equipped with a directional antenna of angle $\alpha$ and range $r$. The coverage
area of the antenna at point $p$ is a circular sector of angle $\alpha$ and radius $r$, whose orientation can be adjusted.
For a given assignment of orientations,
the induced {\em symmetric communication graph} (SCG) of $P$ is the undirected graph, in which two vertices (i.e., points) $u$ and $v$ are connected by an edge if and only if $v$ lies in $u$'s sector and vice versa. In this paper we ask what is the smallest angle $\alpha$ for which there exists an integer $n = n(\alpha)$, such that for any set $P$ of $n$ antennas of angle $\alpha$ and unbounded range, one can orient the antennas so that (i) the induced SCG is connected, and (ii) the union of the corresponding wedges is the entire plane. We show (by construction) that the answer to this problem is $\alpha = \pi/2$, for which $n = 4$. Moreover, we prove that if $Q_1$ and $Q_2$ are two quadruplets of
antennas of angle $\pi/2$ and unbounded range, separated by a line, to which one applies the above construction, independently, then
the induced SCG of $Q_1 \cup Q_2$ is connected. This latter result enables us to apply the construction locally, and to solve the following two further problems.

In the first problem ({\em replacing omni-directional antennas with directional antennas}), we are given a connected unit disk graph, corresponding to a set $P$ of omni-directional antennas of range 1, and the goal is to replace the omni-directional antennas by directional antennas of angle $\pi/2$ and range $r=O(1)$ and to orient them, such that the induced SCG is connected, and, moreover, is an $O(1)$-spanner of the unit disk graph, w.r.t. hop distance. In our solution $r = 14\sqrt{2}$ and the spanning ratio is 8. In the second problem ({\em orientation and power assignment}), we are given a set $P$ of directional antennas of angle $\pi/2$ and adjustable range. The goal is to assign to each antenna $p$, an orientation and a range $r_p$, such that the resulting SCG is (i) connected, and (ii) $\sum_{p \in P} r_p^\beta$ is minimized, where $\beta \ge 1$ is a constant. For this problem, we present an $O(1)$-approximation algorithm.

\end{abstract}
\newpage

\section {Introduction}

Let $P$ be a set of points in the plane, and assume that each point represents a transceiver equipped with a directional antenna. The {\em coverage area} of a directional antenna located at point $p$ of angle $\alpha$ and range $r$, is a sector of angle $\alpha$ of the disk of radius $r$ centered at $p$, where the orientation of the sector can be adjusted. We denote the coverage area of the antenna at $p$ by $\wedge{p}$ (since when assuming unbounded range the sector becomes a wedge).
The induced {\em symmetric communication graph} (SCG) of $P$ is the undirected graph over $P$, in which two vertices (i.e., points) $u$ and $v$ are connected by an edge if and only if $v \in \wedge{u}$ and $u \in \wedge{v}$.

The vast majority of the papers dealing with algorithmic problems motivated by wireless networks, consider omni-directional antennas, whose coverage area is often modeled by a disk. Many of these papers study problems, in which one has to assign radii (under some restrictions) to the underlying antennas so as to satisfy various coverage or communication requirements, while optimizing some measure, such as total power consumption.
Only very recently, researches have begun to study such problems for directional antennas. Directional antennas
have some noticeable advantages over omni-directional antennas. In particular, they require less energy to reach a point at a given distance, and they often reduce the level of interferences in the network.

In this paper we ask the following question:
\begin{problem}\label{prob:smallest_angle}
What is the smallest angle $\alpha$ for which there exists an integer $n = n(\alpha)$, such that for any set $P$ of $n$ points in the plane,
representing transceivers equipped with directional antennas of angle $\alpha$ and unbounded range, one can orient the antennas so that (i) the induced SCG is connected, and (ii) the union of the corresponding wedges is the entire plane, i.e., for any point $x \in \mathbb{R}^2$, there exists a point $p \in P$, such that $x \in \wedge{p}$.
\end{problem}

We would like to use the solution to this problem as a building block in the study of the following two important applications.
These applications have been studied under the asymmetric model of communication (where there is a directed edge from $u$ to $v$ if and only if $v \in W_u$), but not under the (more natural) symmetric model of communication, where they are considerably more difficult.

\vspace{-2mm}
\paragraph{Replacing omni-directional antennas with directional antennas.}
Given a set $P$ of points in the plane, let $\UDG(P)$ be the {\em unit disk graph} of $P$ (i.e., two points of $P$ are connected by an edge if and only if the distance between them is at most 1), and assume that $\UDG(P)$ is connected. Notice that $\UDG(P)$ is the communication graph obtained by placing at each point of $P$ an omni-directional antenna of range 1.
The goal is to replace the omni-directional antennas with directional antennas of some small angle $\alpha$ and range $r$, such that (i) $r=O(1)$, (ii) the induced SCG is connected, and, moreover, (iii) the SCG is an $O(1)$-spanner of $\UDG(P)$, w.r.t. hop distance (i.e., there exists a constant $t \ge 1$, such that, for each edge $(p,q)$ of $\UDG(P)$, there is a path between $p$ and $q$ in the SCG, consisting of at most $t$ hops).

\vspace{-2mm}
\paragraph{Orientation and power assignment.}
Given a set $P$ of directional antennas of angle $\alpha$ and adjustable range, the goal is to assign to each antenna $p$, an orientation and a range $r_p$, such that the resulting SCG is (i) connected, and (ii) $\sum_{p \in P} r_p^\beta$ is minimized, where $\beta \ge 1$ is the distance-power gradient (typically between $2$ and $5$).

\paragraph{Related work.}
A major challenge in the context of directional antennas is how to replace omni-directional antennas with directional antennas, such that (strong) connectivity is preserved, as well as other desirable properties, e.g., short range, similar hop distance, etc.
Several papers have considered this problem under the asymmetric model.
Caragiannis et al.~\cite{CKK+08} consider the problem of orienting the directional antennas and fixing their range, such that the induced graph is strongly connected and the assigned range is minimized.
They present a 3-approximation algorithm for any angle $\alpha \geq 0$; the maximum hop distance in their construction can be linear.
In their survey chapter, Kranakis et al.~\cite{KKM} consider this problem in a more general setting, where each transceiver is equipped with $k \geq 1$ directional antennas.
Damian and Flatland~\cite{DF10} show how to minimize both the range and the hop-ratio (w.r.t. the unit disk graph), for $\alpha \geq \pi/2$. Subsequently, Bose et al.~\cite{BCDFKM11} show how to do it for any $\alpha >0$.
Carmi et al.~\cite{CKLR09} were the first to study this problem under the symmetric model. They show that it is always possible to obtain a connected graph for $\alpha \ge \pi/3$, assuming the range is unbounded (i.e., equal to the diameter of the underlying point set). Later, a somewhat simpler construction was proposed by Ackerman et al.~\cite{AGP10}.
Carmi et al.~\cite{CKLR09} also observe that for $\alpha < \pi/3$ it is not always possible to orient the antennas such that the induced SCG is connected.
Ben-Moshe et al.~\cite{bcckms-dawn-10} investigate the problem of orienting quadrant antennas with only four possible orientations ($\pi/4$, $3\pi/4$, $5\pi/4$, and $7\pi/4$), and vertical half-strip antennas with only two possible orientations (up and down). Both problems are studied under the symmetric model.

The power assignment problem for omni-directional antennas is known to be NP-hard and was studied extensively; see, e.g.,~\cite{KKKP00,CPS99,CMZ02,C10}.
The orientation and power assignment problem, under the asymmetric model, was considered by Nijnatten~\cite{N08}, who observed that there exists a simple $O(1)$-approximation algorithm for any $\alpha \ge 0$. His solution is based on $O(1)$-approximation algorithms for the energy-efficient traveling salesman tour problem. The quality of his approximation does not depend on $\alpha$. Notice that according to the observation of Carmi et al.~\cite{CKLR09} above, there does not always exist a solution to the problem under the symmetric model when $\alpha < \pi/3$.

\paragraph{Our results.}
In Section~\ref{sec:connected_coverage} we show that the solution to problem~\ref{prob:smallest_angle} is $\alpha=\pi/2$, for which $n=4$. That is, we show how to orient any four antennas of angle $\pi/2$, such that there is a path between any two of them and they collectively cover the entire plane (assuming unbounded range).
In order to use this construction as a building block in the solution of appropriate optimization problems, we need to be able to apply it locally, within a small geographic region, and to have a connection between nearby regions. Unfortunately, in Section~\ref{sec:A_cub_B} we give an example showing that we may not have such a connection. We overcome this difficulty by proving the following theorem. If $Q_1$ and $Q_2$ are two quadruplets of
antennas of angle $\pi/2$ and unbounded range, separated by a line, to which one applies the above construction, independently, then
the induced SCG of $Q_1 \cup Q_2$ is connected.
In Section~\ref{sec:replacing_omni} we address the first application above. We show how to replace omni-directional antennas of range 1 with directional antennas of angle $\pi/2$ and range $14\sqrt{2}$ and orient them, such that the induced SCG is an 8-spanner of the unit disk graph, w.r.t. hop distance.
In Section~\ref{sec:power_assignment} we study the second application. We show how to assign an orientation and range to each antenna in a given set of directional antennas of angle $\pi/2$, such that the induced SCG is connected and the total power consumption is at most some constant times the total power consumption in an optimal solution.

\section{Connected coverage of the plane}\label{sec:connected_coverage}

In this section we consider Problem~1.1.
What is the smallest angle $\alpha$ with the property that there exists a positive integer $n=n(\alpha)$,
such that for any set $P$ of at least $n$ points,
one can place directional antennas of angle $\alpha$ and unbounded range at the points of $P$, so that
(i)~the induced SCG is connected,
and (ii)~the plane is entirely covered by the antennas.

We show that the answer to the above question is $\alpha = \pi /2$.
We first show that for $\alpha = \pi/2$, $n=4$ is such an integer. Then, we show that for any $\alpha < \pi/2$, such an integer $n$ does not exist.

\paragraph{Notation.}
We denote the antenna at point $p$ by $\wedge{p}$,
the left ray bounding $\wedge{p}$ (when looking from $p$ into $\wedge{p}$) by $\leftray{p}$,
and the right ray by $\rightray{p}$.
The lines containing these two rays are denoted by $\leftline{p}$ and $\rightline{p}$, respectively.

\subsection {$\alpha=\pi/2$}

\begin{theorem} \label{thm:fourpoints}
Let $P$ be a set of four points in the plane representing the locations of four transceivers equipped with directional antennas of angle $\pi/2$.
Then, one can assign orientations to the antennas, such that
the induced SCG is connected,
and the plane is entirely covered by the four corresponding (unbounded) wedges.
\end{theorem}

\begin{proof}
Denote the convex hull of $P$ by $CH(P)$.
We distinguish between the case where $CH(P)$ is a convex quadrilateral and the case where it is a triangle.
If $CH(P)$ is a convex quadrilateral, then one of its angles is of size at most $\pi/2$.
Each of the two diagonals of $CH(P)$ divides each of its corresponding two angles into two smaller angles,
such that at least one of these smaller angles is of size at most $\pi/2$.
Thus, at least $5$ of the $8$ angles defined by $CH(P)$ and its two diagonals are of size at most $\pi/2$.
Denote the intersection point of the two diagonals by $o$. Then, there exist two adjacent vertices $a,b$ of $CH(P)$, such that
$\angle oab \leq \pi/2$ and $\angle oba \leq \pi/2$. Therefore, one can orient the antennas, such that
the resulting SCG includes the two diagonals and the edge $(a,b)$, and is thus connected.

\begin{figure}[htb]
 \centering
 \subfigure[]{
   \centering
       \includegraphics[width=0.40\textwidth]{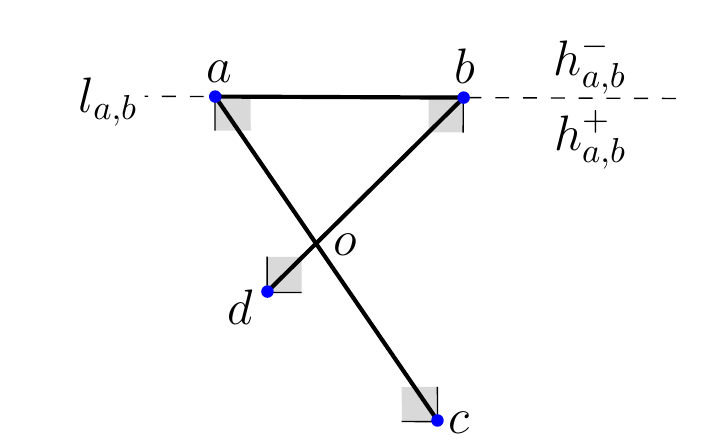}
		   \label{fig:four_pts_convex}
  }
 \subfigure[]{
    \centering
        \includegraphics[width=0.40\textwidth]{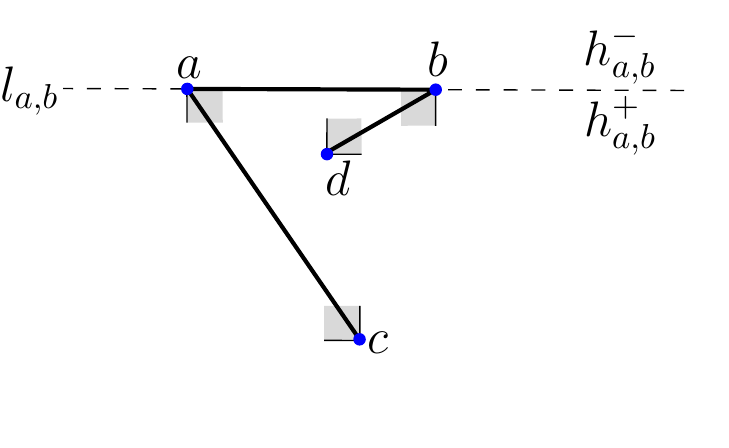}
        \label{fig:four_pts_concave}
 }
	\caption{Proof of Theorem~\ref{thm:fourpoints}.}
\end{figure}

\old{
\begin{figure}[htp]
   \centering
       \includegraphics[width=0.45\textwidth]{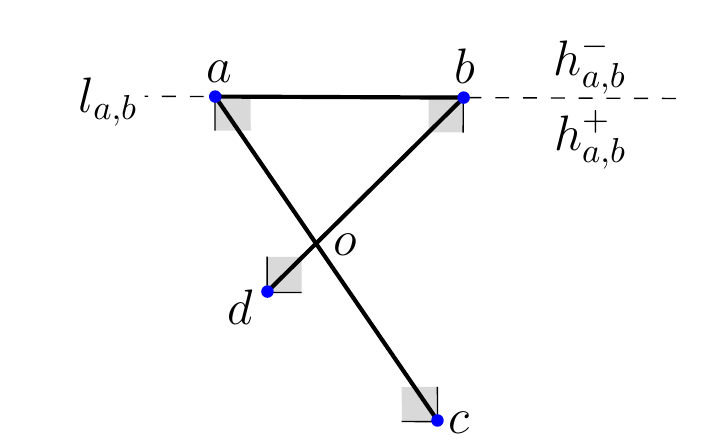}
       \includegraphics[width=0.45\textwidth]{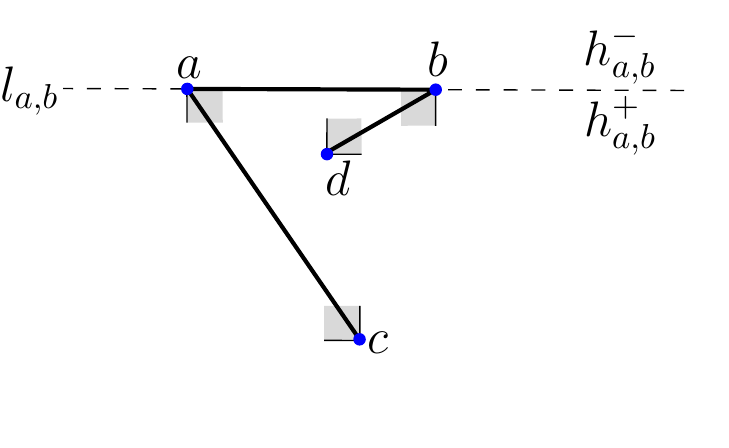}
   \caption{Proof of Theorem~\ref{thm:fourpoints}.}
   \label{fig:four_pts}
\end{figure}
}
Let $c$ and $d$ be the other two vertices of $CH(P)$, such that $c$ is adjacent to $b$.
Then, we saw that one can orient the antennas, such that the resulting SCG includes
the edges $(a,c)$ and $(b,d)$ and the edge $(a,b)$.
Denote by $l_{a,b}$ the line that passes through $a$ and $b$, and
by $h^+_{a,b}$ the closed half plane defined by $l_{a,b}$ and containing $CH(P)$.
Reorient the antenna $\wedge{a}$ at $a$ (resp., $\wedge{b}$ at $b$), such that it faces $CH(P)$ and one of its bounding rays
passes through $b$ (resp., $a$); see Figure~\ref{fig:four_pts_convex}.
Notice that by doing so, we do not lose any of the graph edges, and, moreover,
the half plane $h^+_{a,b}$ is entirely covered by the two antennas.
To complete the proof we need to adjust the antennas at $c$ and $d$, so that the half plane $h^-_{a,b}$
(on the other side of $l_{a,b}$) is also covered.
Since $c$ (resp. $d$) is covered by $\wedge{a}$ (resp., $\wedge{b}$), we reorient $\wedge{c}$ (resp., $\wedge{d}$) so that
it is opposite $\wedge{a}$ (resp., $\wedge{b}$); see Figure~\ref{fig:four_pts_convex}. By doing so,
we do not lose any of the graph edges, and, moreover, the half plane $h^-_{a,b}$ is covered by $\wedge{c} \cup \wedge{d}$.

Assume now that $CH(P)$ is a triangle $\Delta abc$ and that $d \in \Delta abc$. Then, $\Delta abc$ has at least
two angles of size at most $\pi/2$. W.l.o.g., assume $\angle cab \leq \pi/2$ and $\angle cba \leq \pi/2$.
Orient $\wedge{a}$ and $\wedge{b}$, as above, so that $h^+_{a,b}$ is covered by $\wedge{a} \cup \wedge{b}$.
Notice that both $\wedge{a}$ and $\wedge{b}$ contain $\Delta abc$, and therefore both cover $c$ and $d$.
Orient $\wedge{c}$ and $\wedge{d}$, so that $h^-_{a,b}$ is covered by $\wedge{c} \cup \wedge{d}$. The preceding observation implies that
either $\wedge{c}$ covers $a$ and $\wedge{d}$ covers $b$, or vice versa; see Figure~\ref{fig:four_pts_concave}.
Thus, in any case, the obtained graph is connected.

\end{proof}

Let us observe a few properties of the resulting structure.
First, notice that the orientation of each antenna differs from the orientations of the other three by $\pi/2$, $\pi$, and $3\pi/2$, respectively.
Second, each antenna is {\em coupled} with two of the others; namely, with those whose orientation differs from its own by $\pi/2$ and $3\pi/2$, respectively. For example, in Figure~\ref{fig:four_pts_convex}, $\wedge{c}$ is coupled with $\wedge{b}$ and with $\wedge{d}$.
Notice that each such couple covers a half plane. E.g., $\wedge{c}$ and $\wedge{b}$ cover the appropriate half plane defined by $\rightline{c}$, and $\wedge{c}$ and $\wedge{d}$ cover the appropriate half plane defined by $\rightline{d}$.

\paragraph{Remark.}
Clearly, at least four antennas of angle $\pi/2$ are needed in order to cover the entire plane, so we cannot replace the number four in Theorem~\ref{thm:fourpoints} by a smaller number. However, if we are using antennas of angle $2\pi/3$ (resp., $\pi$), then it is relatively easy to show that three points (resp., two points) are sufficient (i.e., $n(2\pi/3) = 3$ and $n(\pi) = 2$).

\subsection{$\alpha < \pi/2$}
As observed by Carmi et al.~\cite{CKLR09}, if $\alpha < \pi/3$,
then it is not always possible to orient the antennas such that the resulting graph is connected.
(Consider, for example, three antennas located at the vertices of an equilateral triangle.
Each of these antennas can cover at most one of the other two vertices, thus it is impossible to obtain a connected graph in this setting.)

For $\alpha \geq \pi/3$,
Carmi et al.~\cite{CKLR09}, and subsequently Ackerman et al.\cite{AGP10},
showed how to obtain a connected symmetric graph, for any set of antennas of angle $\alpha$.
However, their construction does not ensure that the union of the wedges covers the entire plane.
Actually, it is not always possible to orient a set of antennas with angle $\alpha < \pi/2$, so that the induced graph is connected and the entire plane is covered. (This is true even in the asymmetric model where one only requires strong connectivity, as observed by Bose et al.~\cite{BCDFKM11}.)
\old{
To see this, consider, e.g., a set $P$ of points on a line.
At least one of the corresponding antennas must be oriented such that its wedge is empty of other points of $P$.
See Figure~\ref{fig:counter_example};
in order to cover the plane, one needs, in particular, to cover the point $p$. However, if $p$ is far enough from the line,
then any antenna that covers $p$ cannot cover any other point of $P$ (except for the point at its apex).
Thus, the antenna that covers $p$ is isolated in the resulting SCG.

\begin{figure}[htp]
   \centering
       \includegraphics[width=0.6\textwidth]{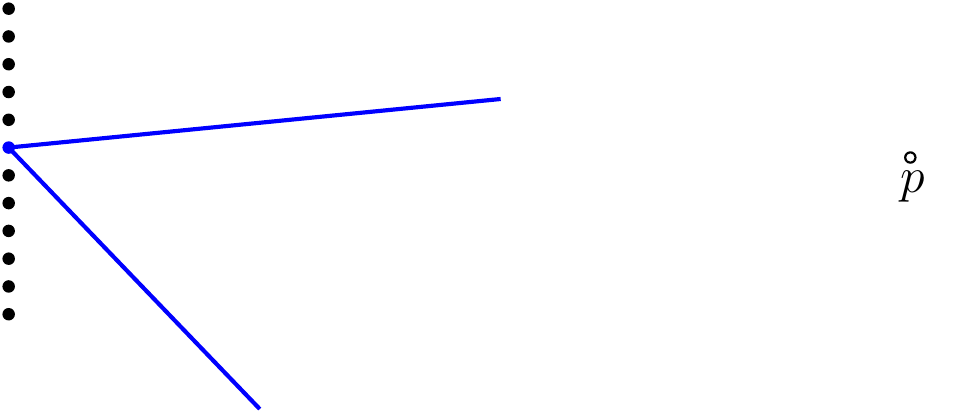}
   \caption{The antenna that covers $p$ is isolated in the resulting SCG.}
   \label{fig:counter_example}
\end{figure}
}
To see this, consider, e.g., a set $P$ of points on a vertical line segment $s$.
In order to cover a point that lies, e.g., far enough to the right of $s$,
at least one of the corresponding antennas must be oriented such that its wedge is empty of points of $P$ (except for the point at its apex).
This antenna is isolated in the resulting SCG.

\section{Separated quadruplets are connected}\label{sec:A_cub_B}

Let $A$ and $B$ be two quadruplets of points (representing transceivers) in the plane, and assume that each of the transceivers
is equipped with a directional antennas of angle $\pi/2$. Orient the antennas corresponding to the points in $A$ (resp., in $B$),
such that they satisfy the conditions of Theorem~\ref{thm:fourpoints}.
Clearly, each point in $A$ is covered by at least one antenna corresponding to a point in $B$, and vice versa.
Unfortunately, this does not imply that the SCG induced by $A \cup B$ is connected;
see Figure~\ref{fig:8pts_no_speartion} for an example where there is no edge between
$A$ and $B$ in the SCG of $A \cup B$.
Theorem~\ref{thm:twosets} below is crucial for our subsequent applications. It states that if the quadruplets $A$ and $B$ can be separated
by a line, then the induced SCG is surely connected.

\begin{figure}[htp]
   \centering
       \includegraphics[width=0.5\textwidth]{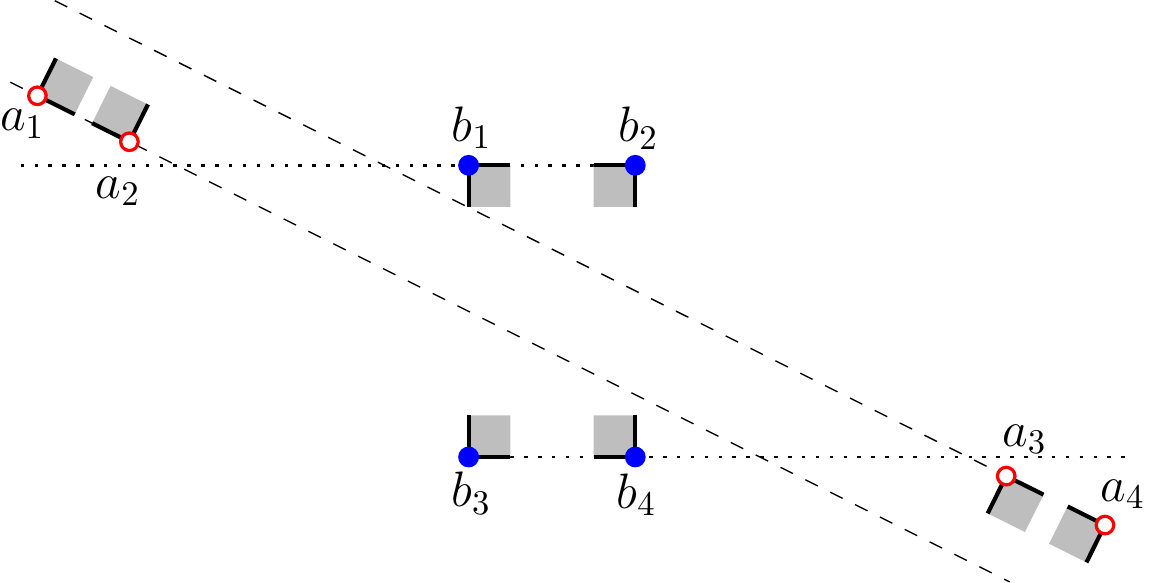}
   \caption{The induced SCG is not connected.}
   \label{fig:8pts_no_speartion}
\end{figure}

\begin{theorem} \label{thm:twosets}
Let $A$, $B$ be two sets of four points each,
and
let the (antennas corresponding to the) points of $A$ and, independently, the points of $B$ be oriented as in the proof of Theorem~\ref{thm:fourpoints}.
If there exists a line $l$ that separates between $A$ and $B$, then
the SCG induced by $A\cup B$ is connected.
\end{theorem}

\begin{proof}
It is enough to show that there exist a point $a \in A$ and a point $b \in B$
that cover each other (i.e., $a \in \wedge{b}$ and $b \in \wedge{a}$).
Assume w.l.o.g. that $l$ is vertical and that the points of $A$ (resp., $B$)
lie to the left (resp., right) of $l$.
Denote by $\halfplane{A}$ (resp., $\halfplane{B}$) the half plane that is defined by $l$ and contains $A$ (resp., $B$).
Let $x_A$ be the smallest number, such that one can pick $x_A$ points of $A$
that (together) cover $\halfplane{B}$. Clearly, $x_A$ is either 2 or 3.
We distinguish between two cases. In the first case, at least one of the two numbers $x_A$ and $x_B$ is 2, where $x_B$ is defined analogously.
In the second case, both $x_A$ and $x_B$ are 3.

\begin{figure}[htb]
 \centering
 \subfigure[Case 1]{
   \centering
       \includegraphics[width=0.4\textwidth]{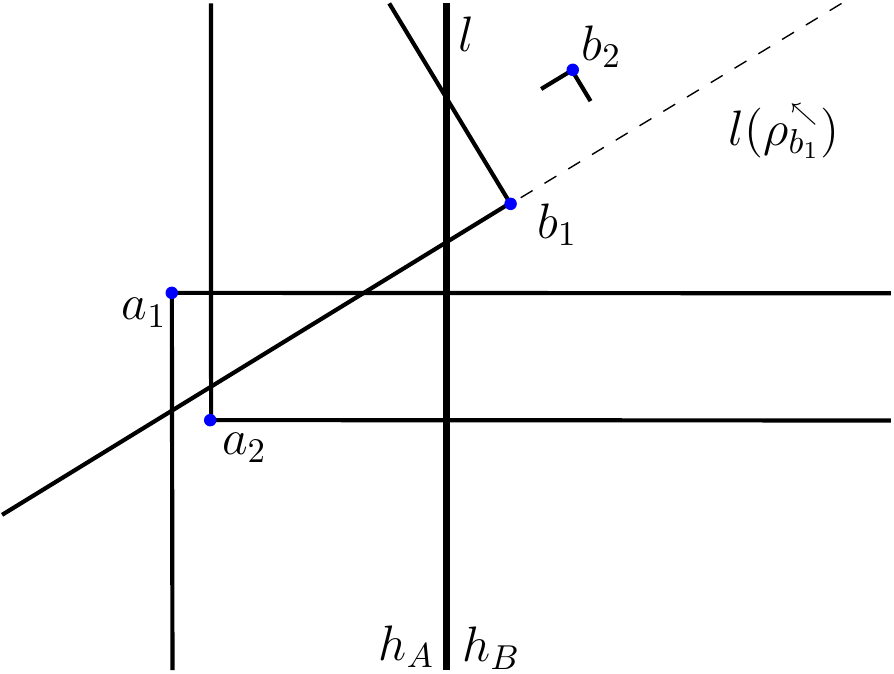}
   		 \label{fig:thm_case1}
  }
 \hspace{0.5cm}
 \subfigure[Case 2]{
    \centering
      \includegraphics[width=0.4\textwidth]{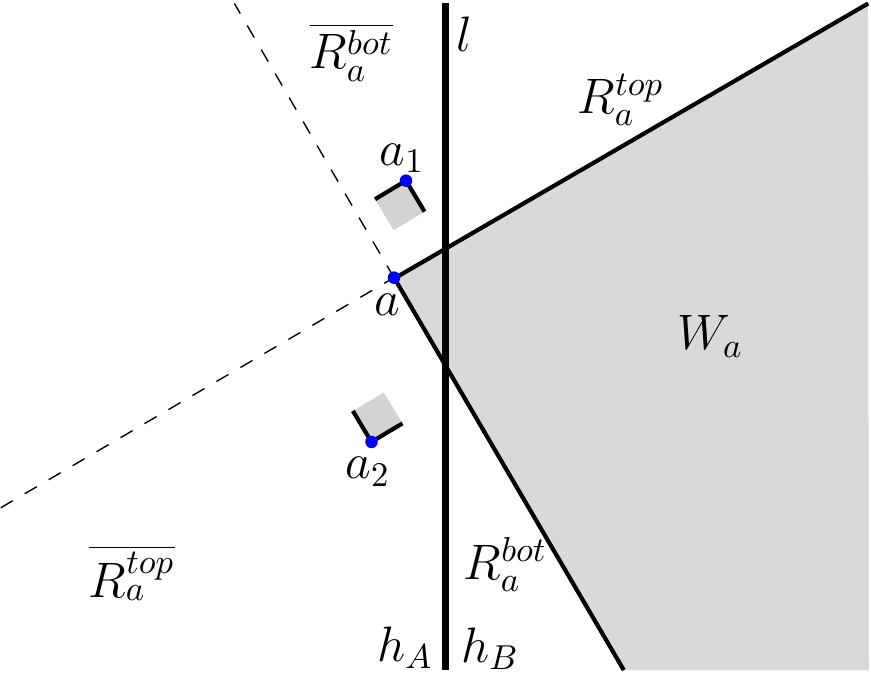}
      \label{fig:thm_case2}
 }
   \caption{Proof of Theorem~\ref{thm:twosets} --- Cases 1 and 2.}
\end{figure}

\old{
\begin{figure}[htp]
   \centering
       \includegraphics[width=0.6\textwidth]{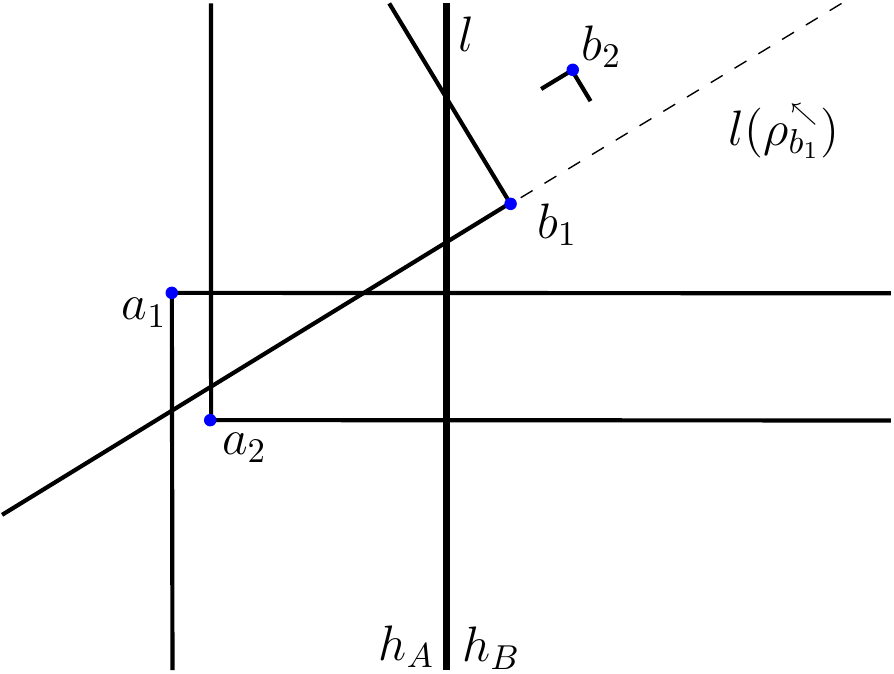}
   \caption{Proof of Theorem~\ref{thm:twosets} --- Case 1.}
   \label{fig:thm_case1}
\end{figure}
}
{\bf Case~1.}
There exist two points from one set that (together) cover the half plane containing the other set.
W.l.o.g., assume points $a_1, a_2 \in A$ cover $\halfplane{B}$, and $a_1$ is not below $a_2$.
In this case, $a_1$ and $a_2$ form a couple (see above), and the half plane covered by them
contains $\halfplane{B}$. That is, the rays $\rightray{a_1}$ and $\leftray{a_2}$ are parallel to $l$, where
$\rightray{a_1}$ is pointing downwards and $\leftray{a_2}$ is pointing upwards, and $\leftray{a_1}$ and $\rightray{a_2}$
are perpendicular to $l$ and pointing rightwards; see Figure~\ref{fig:thm_case1}.
Let $b_1 \in B$ be a point that covers the point $a_1$.
If $b_1$ lies in $\wedge{a_1}$, then we are done, since $a_1$ and $b_1$ cover each other.
Assume, therefore, that $b_1$ lies in $\halfplane{B} \setminus \wedge{a_1} \subseteq \wedge{a_2}$, and that $b_1$ does not
cover $a_2$ (since, otherwise, $a_2$ and $b_1$ cover each other and we are done).
It follows that $\leftline{b_1}$ crosses the line segment $\segment{a_1}{a_2}$.
Let $b_2$ be the point in $B$, such that $b_1$ and $b_2$ form a couple, and $b_2$'s orientation is equal to $b_1$'s orientation plus $\pi/2$ (adding counterclockwise). Then, $b_2$ lies above (or on) $\leftline{b_1}$ (and, of course, $b_2 \in \halfplane{B}$), and $a_2$ lies in $\wedge{b_2}$. But,
the former assertion implies that $b_2$ lies in $\wedge{a_2}$, hence $a_2$ and $b_2$ cover each other, and we are done.

\old{
\begin{figure}[htp]
   \centering
       \includegraphics[width=0.6\textwidth]{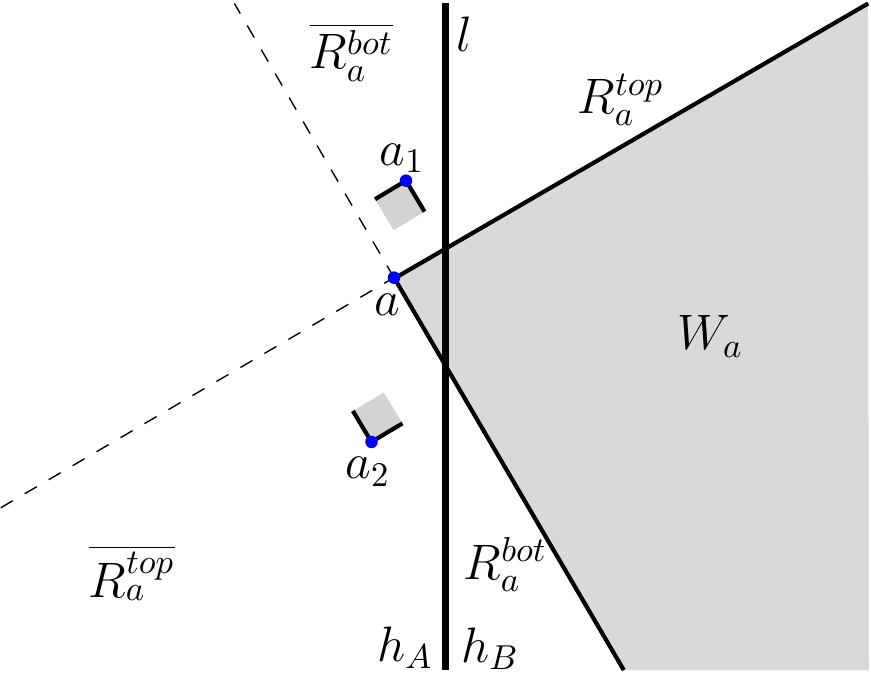}
   \caption{Proof of Theorem~\ref{thm:twosets} --- Case 2.}
   \label{fig:thm_case2}
\end{figure}
}

{\bf Case~2.}
$\halfplane{B}$ is covered by three points of $A$ (but not by two), and $\halfplane{A}$ is covered by three points of $B$ (but not by two).
Notice that on each side of $l$ there exists a point, whose wedge
divides the half plane on the other side of $l$ into three disjoint regions.
More precisely, there exists $a \in A$ that divides $\halfplane{B}$
into three disjoint regions:
$\wedge{a} \cap \halfplane{B}$, $\topregion{a}$, and $\bottomregion{a}$; see Figure~\ref{fig:thm_case2}.
We denote the ``complement'' regions of $\topregion{a}$ and $\bottomregion{a}$ by $\overline{\topregion{a}}$ and $\overline{\bottomregion{a}}$, respectively, where $\overline{\topregion{a}}$ (resp., $\overline{\bottomregion{a}}$) is obtained by rotating $\topregion{a}$ (resp., $\bottomregion{a}$) around its apex by $\pi$.
Notice that $\overline{\topregion{a}} \cap \overline{\bottomregion{a}} \neq \emptyset$.
Let $a_1$ be the point in $A$, such that $a$ and $a_1$ form a couple, and $a_1$'s orientation is equal to $a$'s orientation minus $\pi/2$.
Then, $a_1$ lies above $\rightline{a}$ (i.e., $a_1 \in \overline{\bottomregion{a}}$) and $a_1$ covers $\bottomregion{a}$.
Similarly, let $a_2$ be the point in $A$, such that $a$ and $a_2$ form a couple, and $a_2$'s orientation is equal to $a$'s orientation plus $\pi/2$. Then, $a_2 \in \overline{\topregion{a}}$ and $a_2$ covers $\topregion{a}$; see Figure~\ref{fig:thm_case2} for an illustration.

As above, let $b$ be a point in $B$ that divides $\halfplane{A}$ into three disjoint regions,
$\wedge{b} \cap \halfplane{A}$, $\topregion{b}$, and $\bottomregion{b}$; denote the
``complement'' regions of $\topregion{b}$ and $\bottomregion{b}$ by $\overline{\topregion{b}}$ and $\overline{\bottomregion{b}}$, respectively;
let $b_1, b_2 \in B$,
$b_1 \in \overline{\bottomregion{b}}$ and
$b_2 \in \overline{\topregion{b}}$, such that
$b_1$ covers $\bottomregion{b}$ and $b_2$ covers $\topregion{b}$.

We now show that there exist two points, one in $A$ and one in $B$, that cover each other.
We distinguish between a few subcases.

\noindent
{\bf Case 2a.} $b \in \wedge{a}$ and $a \in \wedge{b}$. That is, $a$ and $b$ cover each other.

\begin{figure}[htb]
 \centering
 \subfigure[Case 2b]{
   \centering
       \includegraphics[width=0.4\textwidth]{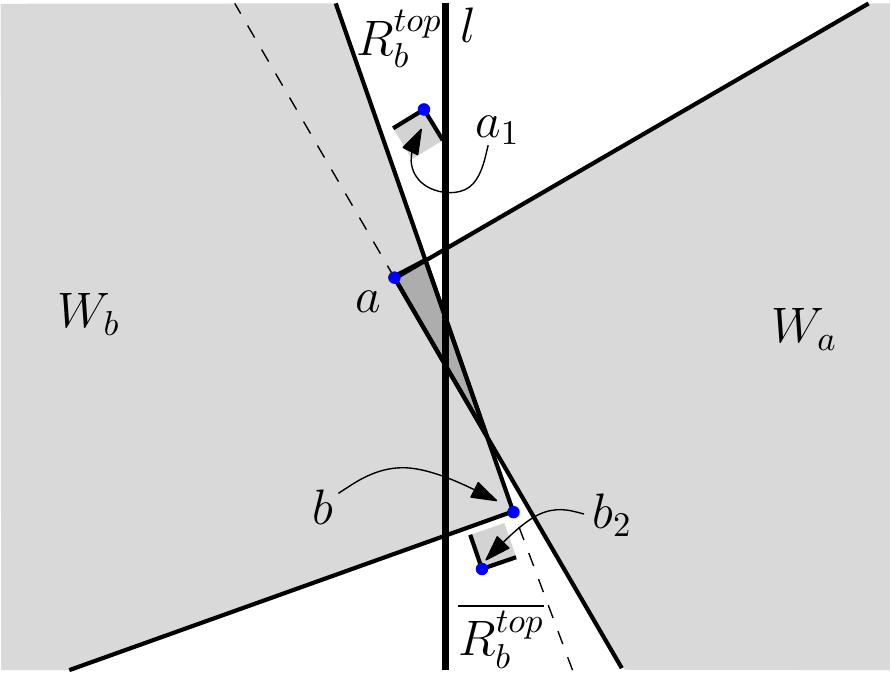}
   \label{fig:thm_case2b}
  }
 \hspace{0.5cm}
 \subfigure[Case 2c]{
     \centering
       \includegraphics[width=0.4\textwidth]{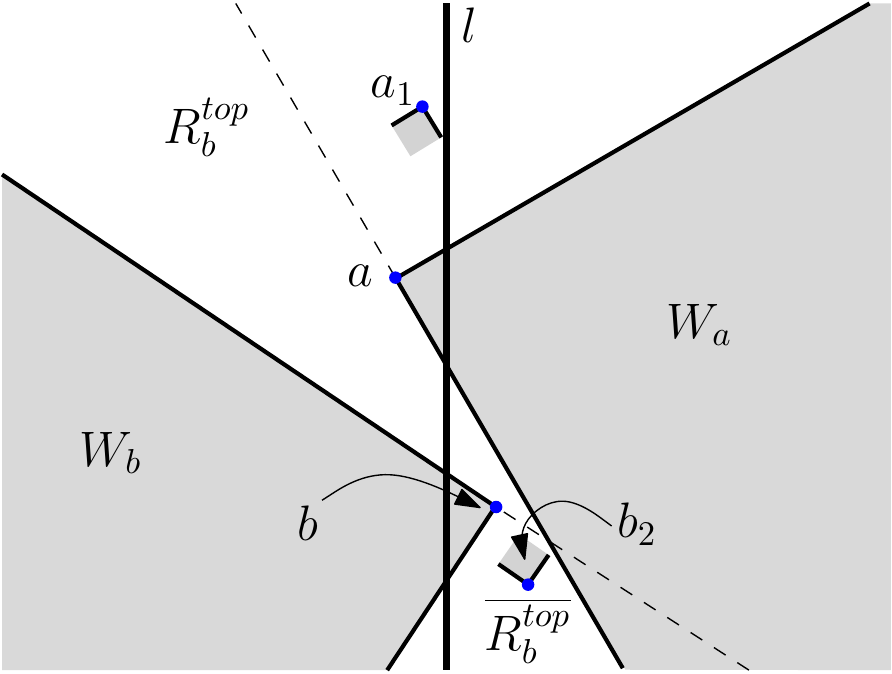}
   \label{fig:thm_case2c}
 }
   \caption{Proof of Theorem~\ref{thm:twosets} --- Cases 2b and 2c.}
\end{figure}

\old{
\begin{figure}[htp]
   \centering
       \includegraphics[width=0.6\textwidth]{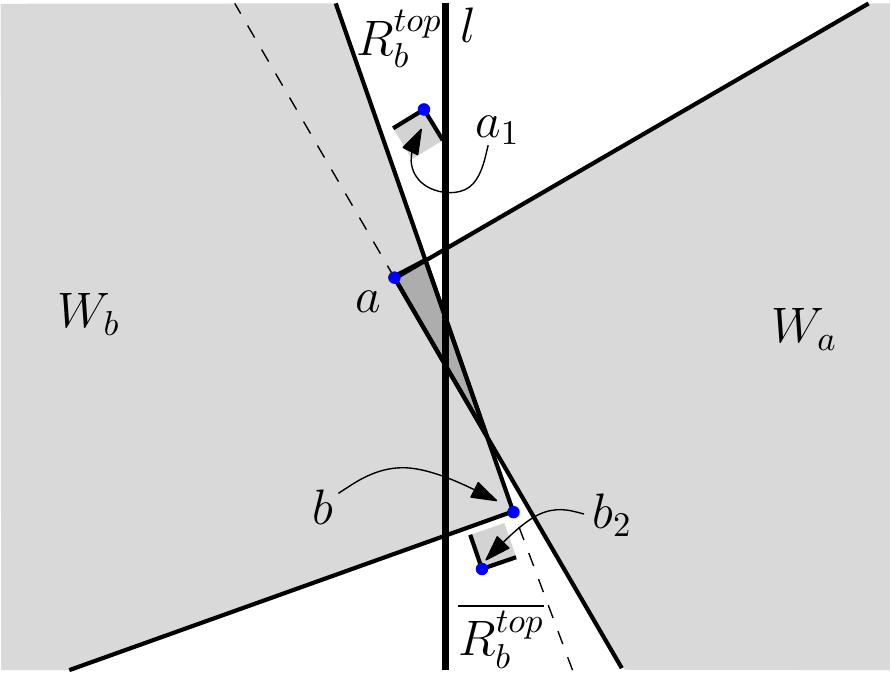}
   \caption{Proof of Theorem~\ref{thm:twosets} --- Case 2b.}
   \label{fig:thm_case2b}
\end{figure}
}

\noindent
{\bf Case 2b.} $b \in \bottomregion{a}$ (that is, $b$ is covered by $a_1$)
and $a \in \wedge{b}$; see Figure~\ref{fig:thm_case2b}.
Assume that $a_1 \notin \wedge{b}$ (otherwise, $a_1$ and $b$ cover each other),
then, $a_1 \in \topregion{b}$ and thus covers $\overline{\topregion{b}}$.
Recall that $b_2 \in \overline{\topregion{b}}$ and covers $\topregion{b}$.
We conclude that $a_1$ and $b_2$ cover each other.

\old{
\begin{figure}[htp]
   \centering
       \includegraphics[width=0.6\textwidth]{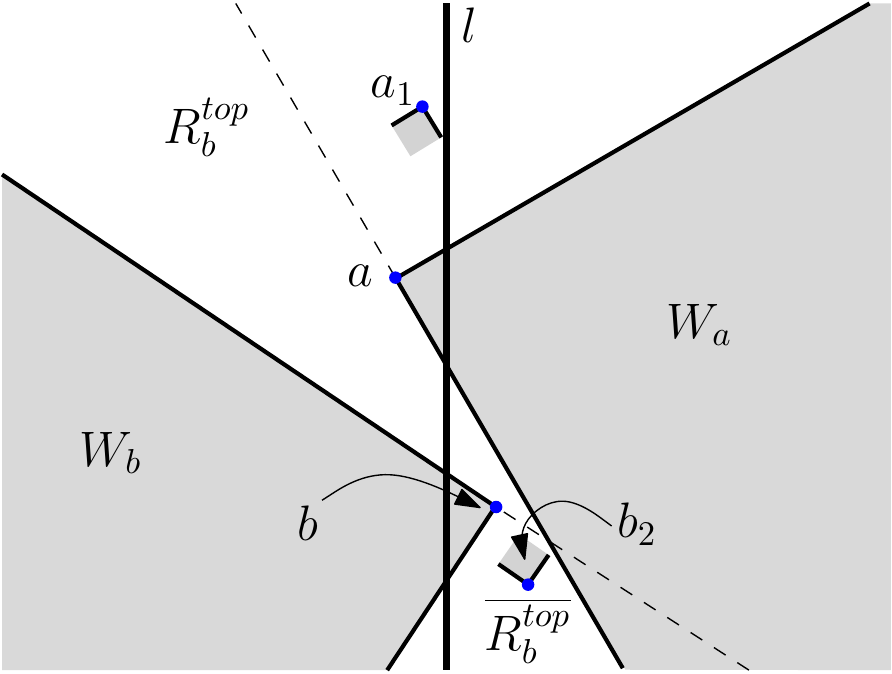}
   \caption{Proof of Theorem~\ref{thm:twosets} --- Case 2c.}
   \label{fig:thm_case2c}
\end{figure}
}

\noindent
{\bf Case 2c.} $b \in \bottomregion{a}$ and $a \notin \wedge{b}$; see Figure~\ref{fig:thm_case2c}.
Then, $a, a_1 \in \topregion{b}$, and, therefore, $b_2$ covers both $a$ and $a_1$.
Now, if $b_2$ lies in $\wedge{a}$, then $a$ and $b_2$ cover each other, and we are done.
Otherwise, $b_2 \in \bottomregion{a}$, but then, $b_2$ and $a_1$ cover each other.

Finally, it is easy to verify that all other subcases are symmetric to either Case~2b or Case~2c.

\end{proof}

\section{Replacing omni-directional antennas with directional antennas}\label{sec:replacing_omni}

Let $P$ be a set of $n$ points in the plane. The {\em unit disk graph} of $P$, denoted $\UDG(P)$, is the graph over $P$, in which there is an edge between two points if and only if the distance between them is at most 1. Notice that $\UDG(P)$ is the communication graph obtained, when each point of $P$ represents a transceiver with an omni-directional antenna of range 1. Assume that $\UDG(P)$ is connected.
Our goal in this section is to replace the omni-directional antennas with directional antennas of angle $\pi/2$ and range $r=O(1)$ and to orient them, such that the induced {\em symmetric} communication graph is (i) connected, and (ii) a $c$-spanner of $\UDG(P)$, with respect to hop distance, where $c$ is an appropriate constant. We first show that this can be done for $r = 14 \sqrt 2$ and $c = 9$.
Then, in Section~\ref{subsec:reducing_hops}, we reduce the hop ratio (i.e., $c$) from $9$ to $8$.

The main idea underlying our construction is to apply Theorem~\ref{thm:fourpoints} multiple times, each time to a cluster of points within a small region, and to use Theorem~\ref{thm:twosets} to establish that the SCG induced by any two such clusters is connected (assuming unbounded range).

We now describe our construction.
Lay a regular grid $\grid$ over $P$, such that the length of a cell side is $7$.
(If $(x_0,y_0)$ is the bottom-left corner of a cell $\C$, then $\C$ is the semi-open square $[x_0, x_0+7) \times [y_0, y_0+7)$.)
For a cell $\C$ of $\grid$, the {\em block} of $\C$ is the $3 \times 3$ portion
of $\grid$ centered at $\C$.
Each of the $8$ cells surrounding $\C$ is a {\em neighbor} of $\C$.
A cell of $\grid$ is considered {\em full} if it contains at least four points of $P$.
It is considered {\em non-full} if it contains at least one and at most three points of $P$.
Proposition~\ref{prop:fullcell} below is analogous to a proposition of Bose et al.~\cite{BCDFKM11}, referring to a similar grid.

\begin{proposition}{(\!\!\cite{BCDFKM11})} \label{prop:fullcell}
Let $\C$ be a cell of $\grid$. Then, any path in $\UDG(P)$ that begins at a point in $\C$
and exits the block of $\C$, must pass through a full cell in $\C$'s block (not including $\C$ itself, which may or may not be full).
In particular, if there are points of $P$ outside $\C$'s block, then at least one of the 8 neighbors of $\C$ is full.
\end{proposition}

\begin{figure}[htp]
   \centering
   \subfigure[]{
       \includegraphics[width=0.45\textwidth]{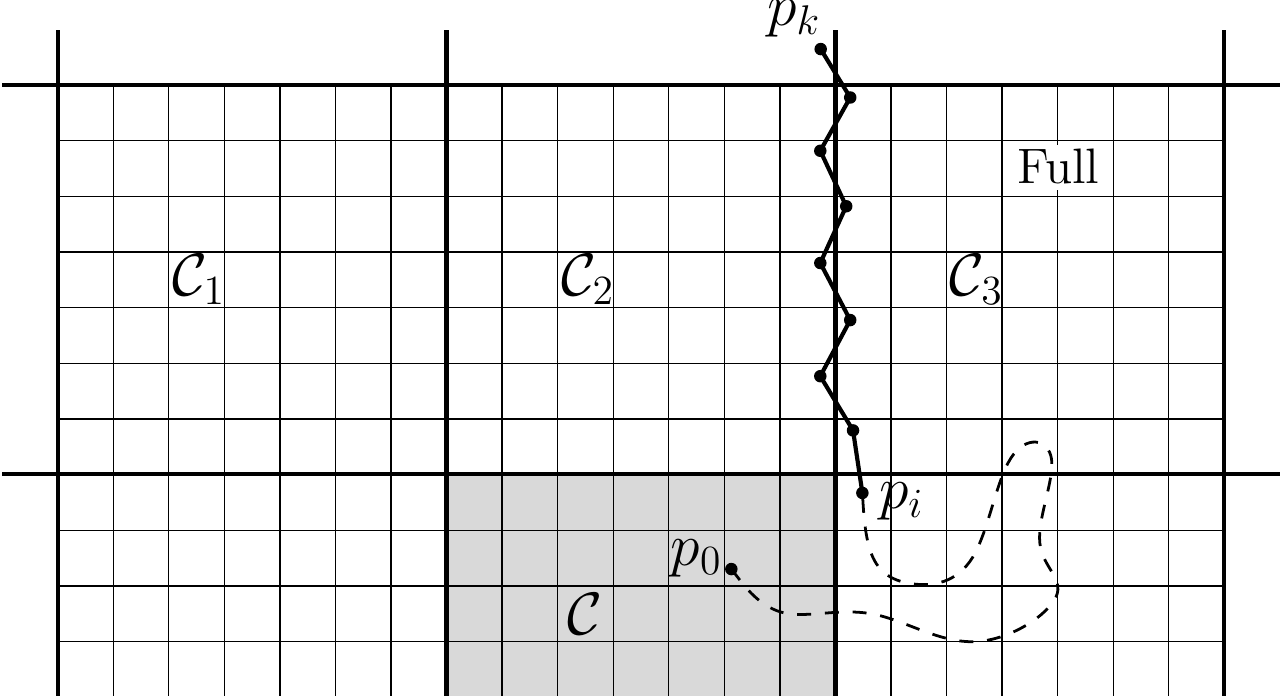}
       }
   \subfigure[]{
       \includegraphics[width=0.45\textwidth]{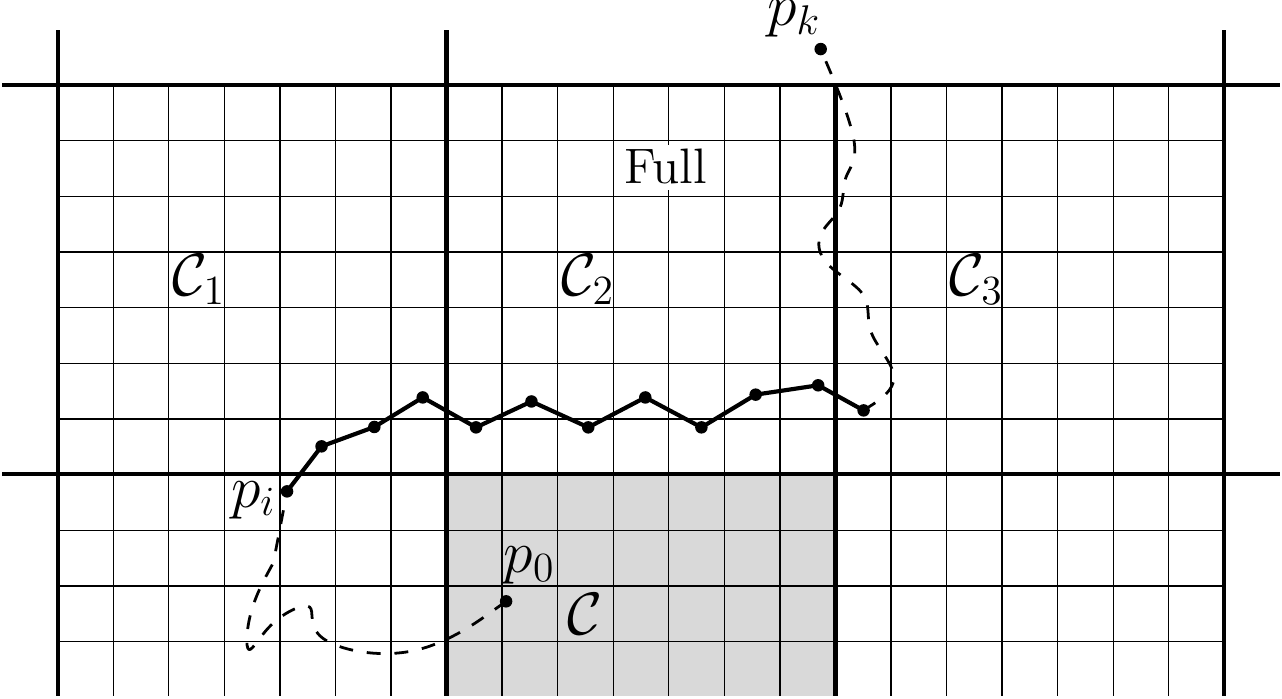}
       }

   \caption{Proposition~\ref{prop:fullcell}.}
   \label{fig:fullcell}
\end{figure}

\old{
\begin{figure}[htp]
   \centering
       \includegraphics[width=0.45\textwidth]{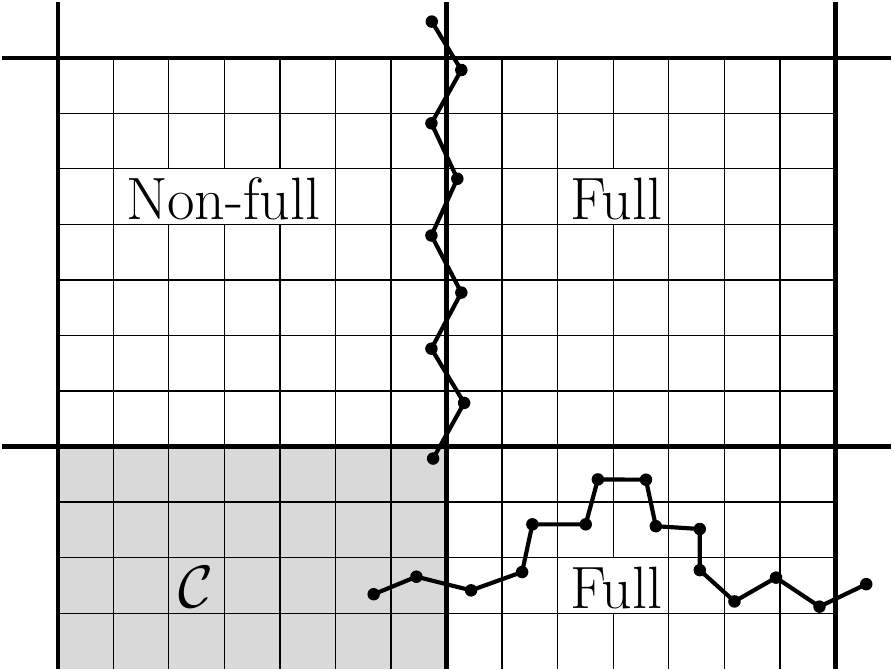}
   \caption{Proposition~\ref{prop:fullcell}.}
   \label{fig:fullcell}
\end{figure}
}

\begin{proof}

\old{
Observe that since each cell is a semi-open square of side length $7$, then any path connecting two non-neighboring cells
contains at least $7$ middle-points.
Now, let $\pi = \left\langle p_0,p_1,\cdots,p_k\right\rangle$ be a path that begins at a point $p_0 \in \C$
and exits the block of $\C$, where $p_k$ is the first point along $\pi$ not in $\C$'s block.
Assume, e.g., that $p_k$ lies above $\C$'s block,
and let $p_i$ be the first point along $\pi$ for which the points $p_i, p_{i+1}, \cdots, p_k$ all lie above the cell $\C$.
By our observation at the beginning of the proof, $1 \leq i \leq 7$ (clearly $i >0$, and $p_{i-1}$ and $p_k$ lies in non-neighboring cells).
Let $\C_1$, $\C_2$, and $\C_3$
be the top cells in the block of $\C$, left to right,
then $p_i, p_{i+1}, \cdots, p_{k-1} \in \C_1 \cup \C_2 \cup \C_3$.
If these (at least $7$) points lies in one or in two of the three cells, then at least one cell contains at least $4$ points, namely is full.
Otherwise, there exists a subpath of $\left\langle p_i, p_{i+1}, \cdots, p_{k-1} \right\rangle$ that connects a point in $\C_1$ with a point in $\C_3$.
By our observation at the beginning of the proof,
such a path contains at least $7$ middle-points lying in $\C_2$, implying that $C_2$ in this case is full.
}

Let $\Pi = \left\langle p_0,p_1,\cdots,p_k\right\rangle$ be a path that begins at a point $p_0 \in \C$
and exits $\C$'s block, where $p_k$ is the first point in $\Pi$ that is not in $\C$'s block.
Assume, e.g., that $p_k$ lies on or above the line containing the top side of $\C$'s block,
and let $p_i$ be the last point in $\Pi$ that lies below the line containing the top side of $\C$; see Figure~\ref{fig:fullcell}.
Then, all points in the subpath
$\left\langle p_i,p_{i+1},\cdots,p_k\right\rangle$ of $\Pi$, except for the two extreme ones, lie in the union of the top three cells of $\C$'s block, and their number (excluding the two extreme points) is at least 7. Let $\C_1$, $\C_2$, and $\C_3$ be the top three cells, from left to right, of $\C$'s block. If $\{p_{i+1},\ldots,p_{k-1}\}$ is contained in only one of these cells, or, alternatively, in only two of these cells, then at least one of these cells contains at least 4 points, and we are done (Figure~\ref{fig:fullcell}(a)). Otherwise, at least one of the points lies in $\C_1$ and at least one of the points lies in $\C_3$, implying that $\C_2$ contains at least 7 points (Figure~\ref{fig:fullcell}(b)).

\end{proof}

If none of the grid cells is full, then, it is easy to see that $P$ is contained in a $14 \times 14$ square. In this case, we can apply Theorem~\ref{thm:fourpoints} to an arbitrary subset $P'$ of four points of $P$, and orient each of the other points of $P$ towards a point of $P'$ that covers it. By setting $r = 14 \sqrt{2}$, we obtain a SCG, in which the hop distance between any two points is at most 5 (see below).
We thus may assume that at least one of the grid cells is full.

For each cell $\C$ of $\grid$,
we orient the points in $\C$ as follows.
If $\C$ is full, then arbitrarily pick four points in $\C$ as $\C$'s {\em hub points},
and orient these hub points according to Theorem~\ref{thm:fourpoints}.
Next, orient each non-hub point in $\C$ towards one of the hub points of $\C$ that covers it.
If $\C$ is non-full, then, for each point $p \in \C$,
pick a full cell closest to $p$,
that is, a full cell containing a point $x_p$, such that the hop distance (in $\UDG(P)$) from $p$ to $x_p$ is not greater than the hop distance from $p$ to any other point of $P$ lying in a full cell.
Denote this cell by $\clst{p}$.
Notice that by Proposition~\ref{prop:fullcell} $\clst{p}$ is a neighbor of $\C$.
Orient each $p \in \C$ towards a hub point of $\clst{p}$ that covers it.
Finally, set $r = 14 \sqrt 2$, so that each antenna can reach any point in its own cell or in a neighboring cell, provided that this point is within the antenna's wedge.

Let $\graph$ be the resulting SCG.
Lemma~\ref{lemma:hop_distance}, below, states that $\graph$ is a $9$-spanner of $\UDG(P)$, w.r.t. hop distance.
In particular, since $\UDG(P)$ is connected, then so is $\graph$.
We first prove the following auxiliary lemma.

\begin{lemma}
\label{lemma:neighboring}
Let $(p,q)$ be an edge of $\UDG(P)$,
and let $\C_p$ and $\C_q$ be the cells of $\grid$, such that $p \in \C_p$ and $q \in \C_q$.
\begin{itemize}
	\item[]
	(i)  	If $\C_p$ and $\C_q$ are both full, then either $\C_p = \C_q$ or $\C_p$ and $\C_q$ are neighbors.
	\item[]
	(ii) 	If $\C_p$ is full and $\C_q$ is non-full, then either $\C_p = \clst{q}$ or $\C_p$ and $\clst{q}$ are neighbors.
	\item[]
	(iii) If $\C_p$ and $\C_q$ are both non-full, then either $\clst{p} = \clst{q}$ or $\clst{p}$ and $\clst{q}$ are neighbors.
\end{itemize}			
\end{lemma}

\begin{proof}
(i) This is obvious, since the Euclidean distance between $p$ and $q$ is at most 1 and the side length of a cell is $7$.
\\
(ii) Let $x_q \in \clst{q}$ be a point that determines the hop distance between $q$ and $\clst{q}$ (see above).
By construction, the hop distance in $\UDG(P)$ from $q$ to $x_q$ is not greater than the hop distance from $q$ to $p$, which is $1$.
Thus, the Euclidean distance between $p \in \C_p$ and $x_q \in \clst{q}$ is at most $2$, whereas the side length of a cell is $7$.
We conclude therefore that either $\C_p = \clst{q}$ or $\C_p$ and $\clst{q}$ are neighbors.
\\
(iii) Assume to the contrary that $\clst{p} \ne \clst{q}$ and that $\clst{p}$ and $\clst{q}$ are not neighbors.
Let $x_p \in \clst{p}$ (resp., $x_q \in \clst{q}$) be a point that determines the hop distance between $p$ and $\clst{p}$ (resp., $q$ and $\clst{q}$).
Consider the path $\Pi$ in $\UDG(P)$ that is obtained by concatenating the shortest path from $x_p$ to $p$, the edge $(p,q)$, and the shortest path from $q$ to $x_q$.
By our assumption,
$\clst{q}$ is not in $\clst{p}$'s block (and vice versa),
thus $\Pi$ starts at $\clst{p}$ and exits its block. By
Proposition~\ref{prop:fullcell},
$\Pi$ passes through a full cell $\C$ in $\clst{p}$'s block,
other than $\clst{p}$ (and other than $\clst{q}$, which is not in $\clst{p}$'s block).
Since $\C$ is full, $\C \neq \C_p,\C_q$. Now, if $\Pi$ visits $\C$ before it visits the point $p$,
then we get that $\C$ is a full cell closer to $p$ than is $\clst{p}$, and, if $\Pi$ visits $\C$ after it visits $q$,
then we get that $\C$ is a full cell closer to $q$ than is $\clst{q}$. Thus,
in both cases we arrive at a contradiction.
\end{proof}

\begin{lemma}
\label{lemma:hop_distance}
$\graph$ is a $9$-spanner of $\UDG(P)$, w.r.t. hop distance.
\end{lemma}

\begin{proof}
Let $(p,q)$ be an edge of $\UDG(P)$.
We show that $\graph$ contains a path from $p$ to $q$ consisting of at most $9$ edges.
Let $\C_p$ and $\C_q$ be the cells of $\grid$, such that $p \in \C_p$ and $q \in \C_q$.
We distinguish between the three cases that are listed in Lemma~\ref{lemma:neighboring}.

\begin{figure}[htb]
 \centering
 \subfigure[]{
    \centering
      \includegraphics[width=0.25\textwidth]{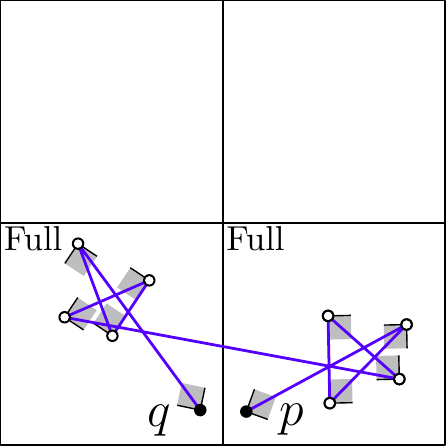}
        \label{fig:hopspanner3}
   }
   \subfigure[]{
     \centering
     \includegraphics[width=0.25\textwidth]{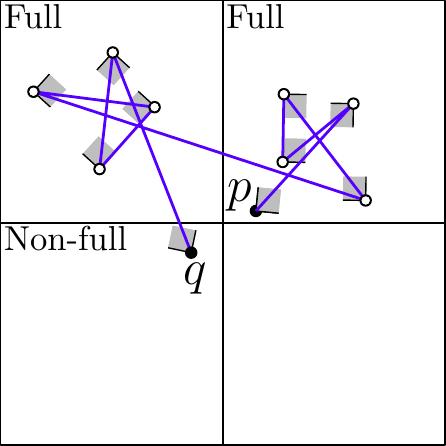}
     \label{fig:hopspanner2}
    }
    \subfigure[]{
     \centering
     \includegraphics[width=0.25\textwidth]{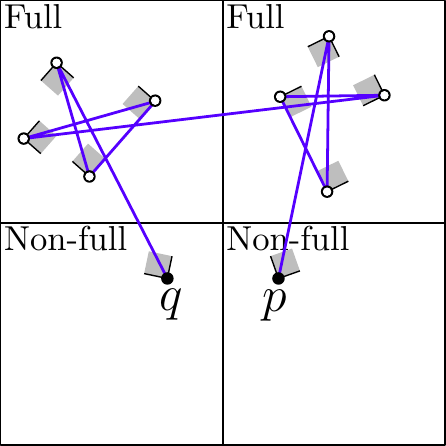}
     \label{fig:hopspanner1}
    }
    \caption{$\graph$ is a $9$-spanner of $\UDG(P)$, w.r.t. hop distance.}
\end{figure}

(i) Consider first the case where $\C_p$ and $\C_q$ are both full.
Then, by Lemma~\ref{lemma:neighboring}, either $\C_p = \C_q$ or $\C_p$ and $\C_q$ are neighbors.
Notice that $p$ is either a hub point of $\C_p$, or it is connected to one by a single edge; and the same holds for $q$ and $\C_q$.
Assume first that $\C_p = \C_q$,
then, there exists a path in $\graph$ from $p$ to $q$ consisting of at most 5 edges.
Indeed, such a path starts at $p$, passes through at most four hub points of $\C_p$,
and ends at $q$.
Assume now that $\C_p$ and $\C_q$ are neighbors.
By Theorem~\ref{thm:twosets} and since the range of each antenna is sufficient to reach any point in any neighboring cell,
there exists an edge in $\graph$ connecting a hub point of $\C_p$ with a hub point of $\C_q$.
Consequently, $\graph$ contains a path from $p$ to $q$ consisting of at most 9 edges.
Indeed, such a path starts at $p$,
passes through at most four hub points of $\C_p$,
continues to a hub point of $\C_q$,
passes through at most four hub points of $\C_q$,
and finally ends at $q$;
see Figure~\ref{fig:hopspanner3} for an illustration.

(ii) Consider now the case where one of the cells, say $\C_p$, is full and the other is non-full.
By construction, $q$ is oriented to a hub point of a neighboring full cell $\clst{q}$ that covers it.
By Lemma~\ref{lemma:neighboring}, either $\C_p = \clst{q}$ or $\C_p$ and $\clst{q}$ are neighbors.
Therefore, as in case (i) above, $\graph$ contains a path from $p$ to $q$ consisting of at most 9 edges;
see Figure~\ref{fig:hopspanner2} for an illustration.

(iii) Finally, consider the case where $\C_p$ and $\C_q$ are both non-full.
By construction, $p$ (resp., $q$) is connected by an edge to a hub point of a full cell $\clst{p}$ (resp., $\clst{q}$) that covers it.
By Lemma~\ref{lemma:neighboring}, either $\clst{p} = \clst{q}$ or $\clst{p}$ and $\clst{q}$ are neighbors.
Therefore, as in case (i) above, $\graph$ contains a path from $p$ to $q$ consisting of at most 9 edges;
see Figure~\ref{fig:hopspanner1} for an illustration.
\end{proof}

\subsection{Reducing the hop ratio to 8}\label{subsec:reducing_hops}

We show how to reduce the hop ratio from 9 to 8, by picking the four hub points (in each full cell) more carefully, and by reorienting the points in non-full cells, so that whenever two such points are ``close'' to each other, they are oriented towards the same full cell.

\begin{figure}[htp]
   \centering
       \includegraphics[width=0.4\textwidth]{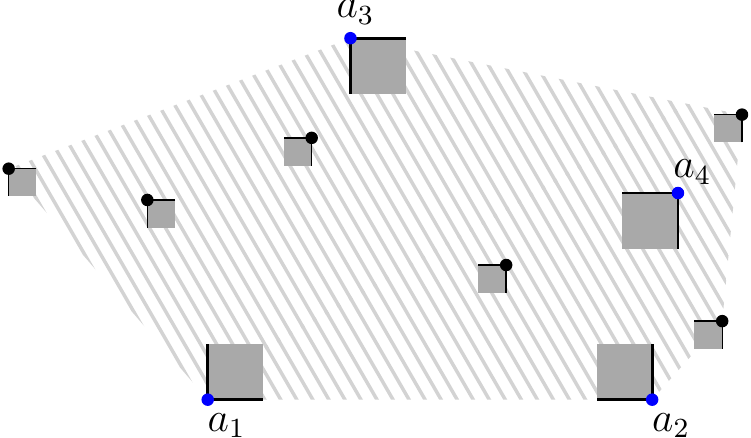}
   \caption{The orientation assignment to $\{a_1, a_2, a_3, a_4\}$ satisfies the conditions of Theorem~\ref{thm:fourpoints},
   and each of the points in $P_\C$ is covered by $a_1$ or $a_2$.}
   \label{fig:four_pts_improved}
\end{figure}

We first show how to carefully pick the four hub points of a full cell; see Figure~\ref{fig:four_pts_improved} for an illustration.
Let $\C$ be a full cell and put $P_\C = P \cap \C$. Let $\segment{a_1}{a_2}$ be the longest edge of $CH(P_\C)$, and assume, w.l.o.g., that $\segment{a_1}{a_2}$ is horizontal and that $CH(P_\C)$ lies above it.
The half-strip with base $\segment{a_1}{a_2}$ and facing upwards
covers at least one point of $P_\C \setminus \{a_1,a_2\}$
(otherwise, $CH(P_\C)$ has an edge longer than $\segment{a_1}{a_2}$). Let $a_3$ be such a point, and let $a_4$ be an arbitrary point in
$P_\C \setminus \{a_1,a_2,a_3\}$. We pick $a_1,a_2,a_3,a_4$ as the hub points of $\C$, where
$a_1$ and $a_2$ are the {\em supporting} hub points of $\C$.
Assume w.l.o.g. that $a_1$ is to the left of $a_2$ and that $a_3$ is to the left of $a_4$.
Assign $a_1$ orientation up-right, $a_2$ orientation up-left, $a_3$ down-right, and $a_4$ down-left.
This orientation assignment satisfies the conditions of Theorem~\ref{thm:fourpoints}.
Moreover, since $CH(P_\C) \subseteq \wedge{a_1} \cup \wedge{a_2}$, we can orient each non-hub point of $\C$
towards one of the supporting hub points of $\C$.

Now, we would like to modify the rule by which we orient the points in non-full cells. For this we need Lemma~\ref{lemma:full_super_neighbor} below.
Let $P_{\nf} \subseteq P$ be the subset of points of $P$ lying in non-full cells, and
consider $\UDG(P_{\nf})$, the subgraph of $\UDG(P)$ induced by $P_{\nf}$.
Let $p$ and $q$ be two points in the same connected component of $\UDG(P_{\nf})$,
and let $\C_p$ and $\C_q$ be the (non-full) cells of $\grid$, such that $p \in \C_p$ and $q \in \C_q$.
As observed above, Proposition~\ref{prop:fullcell} implies that $\clst{p}$ is a neighbor of $\C_p$ and that $\clst{q}$ is a neighbor of $\C_q$.
Lemma~\ref{lemma:full_super_neighbor} below, shows that $\clst{p}$ is also a neighbor of $\C_q$ and that $\clst{q}$ is also a neighbor of $\C_p$.

\begin{lemma}
\label{lemma:full_super_neighbor}
Consider a connected component of $\UDG(P_{\nf})$, and let $P' \subseteq P$ be its corresponding set of points.
Let $p,q$ be two points in $P'$, and let $\C_q$ be the cell of $\grid$ such that $q \in \C_q$.
Then $\clst{p}$ is a neighbor of $\C_q$.
\end{lemma}
\begin{proof}
Assume to the contrary that $\clst{p}$ is not a neighbor of $\C_q$.
Let $x_p \in \clst{p}$ be a point that determines the hop distance between $p$ and $\clst{p}$.
Consider the path $\Pi$ in $\UDG(P)$ that is obtained by concatenating the subpaths $\Pi_1$ and $\Pi_2$,
where $\Pi_1$ is a path in $\UDG(P_{\nf})$ from $q$ to $p$,
and $\Pi_2$ is the shortest path in $\UDG(P)$ from $p$ to $x_p$.
By our assumption, $\Pi$ exits $\C_q$'s block.
Therefore, by Proposition~\ref{prop:fullcell}, $\Pi$ must pass through a full cell in $\C_q$'s block before exiting the block.
By definition, $\Pi_1$ does not pass through full cells, hence $\Pi_2$ must pass through a full cell $\C$ in $\C_q$'s block.
But, $\Pi_2$ visits $\C$ before it visits $\clst{p}$, thus $\C$ is a full cell closer to $p$ than is $\clst{p}$
--- a contradiction.
\end{proof}

\old{
\begin{corollary}
\label{col:super_neighbor_full_cell}
Let $P' \subseteq P$ corresponds to a connected component of $\UDG(\nf)$,
and let $p \in P'$ be an arbitrary point.
Then, for any $q \in P'$, the cell of $\grid$ containing $q$ is a neighbor of $\clst{p}$.
\end{corollary}
}

Given Lemma~\ref{lemma:full_super_neighbor}, we may reorient the points in $P_{\nf}$ as follows.
For each connected component of $\UDG(P_{\nf})$ and its corresponding set of points $P'$, pick any point $p \in P'$ and orient it,
as before, towards a hub point of $\clst{p}$ that covers it. Now, for each point $q \in P'$, $q \ne p$, orient $q$ towards a hub point
of $\clst{p}$ that covers it.

Let $\graph$ be the resulting SCG.
We are now able to reduce the number of hops in each of the paths described in the proof of Lemma~\ref{lemma:hop_distance} above.
Let $(p,q)$ be an edge of $\UDG(P)$, and let $\C_p$ and $\C_q$ be the cells of $\grid$, such that $p \in \C_p$ and $q \in \C_q$.
(i)~Consider first the case where $\C_p$ and $\C_q$ are both full.
Then, there exists a path in $\graph$ from $p$ to $q$ consisting of at most 7 edges.
Indeed, both $p$ and $q$ are connected to a supporting hub point by a single edge,
therefore, in the worst case, a path from $p$ to $q$ starts at $p$, passes through
at most three hub points of $\C_p$,
continues to a hub point of $\C_q$,
passes through at most three hub points of $\C_q$,
and finally ends at $q$.

(ii)~Consider now the case where $\C_p$ is full and $\C_q$ is non-full.
Let $r$ be the point in the connected component of $\UDG(P_{\nf})$ to which $q$ belongs, such that each point in $\UDG(P_{\nf})$ was oriented towards a hub point of $\clst{r}$ that covers it. We claim that $\clst{r}$ is a neighbor of $\C_p$, even though it is possible that $\clst{r} \ne \clst{q}$.
Indeed, assuming that $\clst{r}$ is not in $\C_p$'s block leads to a contradiction, as above. (Define the path in $\UDG(P)$ that is obtained by concatenating the edge $(p,q)$, the path in $\UDG(P_{\nf})$ from $q$ to $r$, and the shortest path in $\UDG(P)$ from $r$ to a point in $\clst{r}$. Then, by Proposition~\ref{prop:fullcell}, there is a full cell closer to $r$ than is $\clst{r}$.)
However, the hub point of $\clst{r}$ towards which $q$ is oriented, is not necessarily a supporting hub point,
thus the path from $p$ to $q$ consists of at most 8 edges.
Such a path starts at $p$,
passes through at most three hub points of $\C_p$,
continues to a hub point of $\clst{r}$, passes through at most four hub points of $\clst{r}$,
and finally ends at $q$.

(iii)~Finally, consider the case where $\C_p$ and $\C_q$ are both non-full.
By construction, $p$ and $q$ are oriented to the same full cell $\C$,
thus the path from $p$ to $q$ starts at $p$,
passes through at most four hub points of $\C$, and finally ends at $q$;
such a path consists of at most 5 edges.

Theorem~\ref{thm:replacing_summary} summarizes the main result of this section.
\begin{theorem}
\label{thm:replacing_summary}
Let $P$ be a set of points, where each point represents a transceiver equipped with an omni-directional antenna of range 1,
and assume that $\UDG(P)$ is connected.
Then, one can replace the omni-directional antennas with directional antennas of angle $\pi/2$ and range $14 \sqrt{2}$,
such that the induced SCG is (i) connected, and (ii) a 8-spanner of $\UDG(P)$, w.r.t. hop distance.
\end{theorem}

\old{
\begin{figure}[htb]
 \centering
 \subfigure[]{
    \centering
      \includegraphics[width=0.25\textwidth]{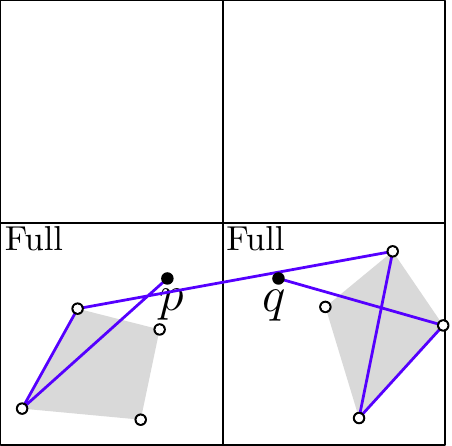}
        \label{fig:hopspanner6}
   }
   \subfigure[]{
     \centering
     \includegraphics[width=0.25\textwidth]{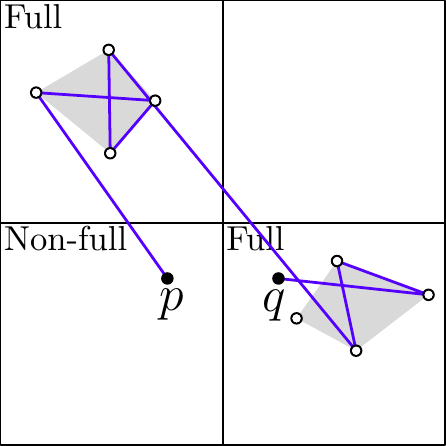}
     \label{fig:hopspanner5}
    }
    \subfigure[]{
     \centering
     \includegraphics[width=0.25\textwidth]{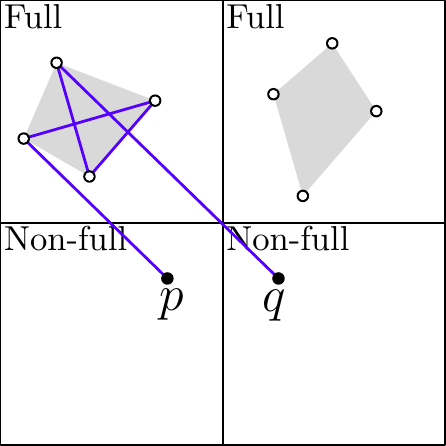}
     \label{fig:hopspanner4}
    }
    \caption{$8$-hop distance}
\end{figure}
}


\section{Orientation and power assignment}\label{sec:power_assignment}

In the construction described in Section~\ref{sec:replacing_omni}, each antenna is assigned range $r=14\sqrt{2}$.
In this section we consider the problem of assigning to each antenna $p$ an individual range, denoted $r(p)$.
Let $P$ be a set of $n$ points in the plane, representing the locations of directional antennas of angle $\pi/2$.
In the {\em orientation and power assignment} problem one needs to assign to each of the antennas an orientation and a range, such that
(i) the resulting SCG is connected, and (ii) $\sum_{p \in P}{r(p)^\beta}$ is minimized, where $\beta \ge 1$ is the distance-power gradient (typically, $2 \leq \beta \leq 5$).

Assume, for convenience, that $n = 8m$, for some integer $m$.
Let $\Pi$ be a simple closed polygonal path whose vertices are the points in $P$.
We partition $P$ into subsets (called {\em sections}) of size 8, by traversing $\Pi$ from an arbitrary point $p \in P$. 
That is, each of the sections consists of eight consecutive points along $\Pi$.
For each section $S$, we partition its points, according to their $x$-coordinate, into a left subsets, $S^l$, consisting of the four leftmost points of $S$, and a right subset, $S^r$, consisting of the four rightmost points of $S$.
(Notice that the points in each quadruplet are not necessarily consecutive along $\Pi$.)
Next, we orient the antennas corresponding to the points in $S^l$ (resp., $S^r$) according to Theorem~\ref{thm:fourpoints}.
By definition (of $S^l,S^r$), there exists a vertical line, $l_S$, that separates between $S^l$ and $S^r$.
Therefore, by Theorem~\ref{thm:twosets}, the SCG induced by $S^l \cup S^r$ is connected.

Moreover, consider two `adjacent' sections $A$ and $B$, and consider their corresponding lines $l_A$ and $l_B$.
Assume, e.g., that $l_A$ lies to the left of $l_B$; see Figure~\ref{fig:power_assignment}. Then, $l_A$ also separates between $A^l$ and $B^r$,
and therefore, assuming unbounded range, the SCG induced by $A \cup B$ is connected.
We conclude that any two adjacent sections are connected, and one can limit the range of the points in a section as long as they reach the points in the preceding and succeeding sections (and each other).

\begin{figure}[htp]
   \centering
       \includegraphics[width=0.45\textwidth]{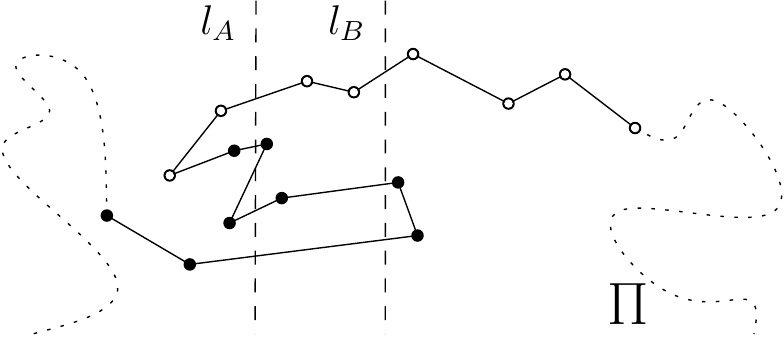}
   \caption{The points of $A$ are shown as filled circles and the points of $B$ as empty circles.}
   \label{fig:power_assignment}
\end{figure}

We take $\Pi$ to be an $O(1)$-approximation of an optimal solution to the traveling salesperson problem in the complete graph
$G_\beta(P)$ over $P$, in which the weight of an edge $e=(p,q)$ is $|e|^\beta$; $\Pi$ can be computed in polynomial time, see, e.g., \cite{BC00,BNSWW10,FLNL08,A01}.
For each section $S_i$, $1 \le i \le m$, and for each point $p \in S_i$, we set $r(p)$ to be the maximum distance between $p$ and a point in $S_{i-1} \cup S_i \cup S_{i+1}$,
where $S_{i-1}, S_{i+1}$ are the preceding and succeeding sections of $S_i$, respectively ($S_0=S_m$ and $S_{m+1}=S_1$).
By construction, the induced SCG is connected.

Below we analyze the quality of the obtained power assignment, denoted $PA$, w.r.t. to the optimal power assignment using omni-directional antennas, denoted $OPT$. Recall that the cost of $PA$ (denoted $Cost(PA)$) is $\sum_{p \in P} r(p)^\beta$.
Clearly, $Cost(OPT)$ is not greater than the cost of an optimal power assignment using directional antennas of angle $\pi/2$.

\begin{theorem}
$Cost(PA) = O(1) \cdot Cost(OPT)$.
\end{theorem}

\begin{proof}
Let $S_i$ be the $i$'th section. Then,
\begin{align*}
Cost(S_i) &= \sum_{p \in S_i} r(p)^\beta \le 8 \cdot \max \{ |pq|^\beta \, : \, p \in S_i \mbox{ and } q \in S_{i-1} \cup S_i \cup S_{i+1} \} \\
& \le 8 \cdot 15^\beta \cdot \max \{ |pq|^\beta \, : \, p, q \in S_{i-1} \cup S_i \cup S_{i+1} \mbox{ and are consecutive in } \Pi \} \, .
\end{align*}
Let $MST(G_\beta(P))$ denote a minimum spanning tree of the graph $G_\beta(P)$ (defined above), and
let $TSP(G_\beta(P))$ denote a minimum tour of $G_\beta(P)$. It is well known and easy to prove that
$Cost(TSP(G_\beta(P))) \le O(1) \cdot Cost(MST(G_\beta(P)))$, where $Cost()$ is the sum of the weights of the edges in the appropriate structure.
Kirousis et al.~\cite{KKKP00} argued that $Cost(MST(G_\beta(P))) \le Cost(OPT)$. It follows that $Cost(TSP(G_\beta(P))) \le O(1) \cdot Cost(OPT)$.
Thus,
\begin{align*}
Cost(PA) &= \sum_{i=1,m}{Cost(S_i)} \le
\sum_{i=1,m} 8 \cdot 15^\beta \cdot \max \{ |pq|^\beta \, : \, p, q \in S_{i-1} \cup S_i \cup S_{i+1} \mbox{ and are consecutive in } \Pi \} \\
 & \leq 8 \cdot 15^\beta \cdot 3 \cdot \sum_{e \in \Pi} |e|^\beta =  O(1) \cdot Cost(\Pi) = O(1) \cdot Cost(TSP(G_\beta(P))) = O(1) \cdot Cost(OPT) \, .
\end{align*}

\end{proof}

\paragraph{Remark.}
Our goal in this section was to establish that there exists an $O(1)$-approximation algorithm for the orientation and power assignment problem. We did not try to optimize the constant.



\end{document}